\documentclass[a4paper,USenglish]{lipics}
 
\usepackage{microtype}
\usepackage{tikz}
\usetikzlibrary{positioning,shapes,chains,backgrounds,arrows,chains,fit,decorations}
\usepackage[T1]{fontenc}
\usepackage{graphicx}
\usepackage{amsmath}
\usepackage{enumerate}
\usetikzlibrary{decorations.pathreplacing}

\newcommand{\dom}{\mathrm{dom}}
\newcommand{\extdom}{\mathrm{extdom}}
\theoremstyle{definition}

\title{On the Size of Lempel-Ziv and Lyndon Factorizations}
\author[1]{Juha K\"{a}rkk\"{a}inen}
\author[1]{Dominik Kempa}
\author[2]{Yuto Nakashima}
\author[1]{Simon~J.~Puglisi}	
\author[3]{Arseny M. Shur}
\affil[1]{Helsinki Institute for Information Technology (HIIT) and\\
Department of Computer Science, University of Helsinki, Finland\\
  \texttt{\{juha.karkkainen,dominik.kempa,simon.puglisi\}@cs.helsinki.fi}}
\affil[2]{Department of Informatics,
Kyushu University, Japan and\\ Japan Society for the Promotion of
Science, Japan\\
  \texttt{yuto.nakashima@inf.kyushu-u.ac.jp}}
\affil[3]{Department of Algebra and Discrete Mathematics,\\
Ural Federal University, Russia\\
  \texttt{arseny.shur@urfu.ru}}
\authorrunning{J. K\"{a}rkk\"{a}inen, D. Kempa, Y. Nakashima, S. J. Puglisi, and A. M. Shur}
\Copyright{Juha K\"{a}rkk\"{a}inen, Dominik Kempa, Yuto Nakashima, Simon J. Puglisi, and Arseny M. Shur}
\serieslogo{}
\volumeinfo
  {Billy Editor and Bill Editors}
  {2}
  {Conference title on which this volume is based on}
  {1}
  {1}
  {1}
\EventShortName{}
\DOI{10.4230/LIPIcs.xxx.yyy.p}

\begin{document}

\maketitle

\begin{abstract}
  Lyndon factorization and Lempel-Ziv (LZ) factorization are both
  important tools for analysing the structure and complexity of
  strings, but their combinatorial structure is very different. In
  this paper, we establish the first direct connection between the two
  by showing that while the Lyndon factorization can be bigger than
  the non-overlapping LZ factorization (which we demonstrate by
  describing a new, non-trivial family of strings) it is never more
  than twice the size.
\end{abstract}

\section{Introduction}

Given a string (or word) $x$, a {\em factorization} of $x$ partitions
$x$ into substrings $f_1,f_2, \ldots f_t$, such that $x = f_1f_2\ldots
f_t$. In the past 50 years or so, dozens of string factorizations have
been studied, some purely out of combinatorial interest
(e.g.~\cite{FGKK14,BBGIIIPS16,MIBTM16,BannaiGIKKPPS15}) and others because the
internal structure that they reveal allows the design of efficient
string processing algorithms. Perhaps the two most important
factorizations in string processing are the Lempel-Ziv (LZ)
factorization~\cite{ZL77} and the Lyndon factorization\footnote{Also
  known as the Standard factorization.}~\cite{CFL58}.

The LZ factorization has its origins in data compression, and is still
used in popular file compressors\footnote{For example {\tt gzip}, {\tt
    p7zip}, {\tt lz4}, and {\tt snappy} all have the LZ factorization
  at their core.}  and as part of larger software systems (see,
e.g.,~\cite{CDGHWBCFG08,HPZ11} and references therein).  More recently
it has been used in the design of compressed data structures for
indexed pattern matching~\cite{GGKNP14} and other
problems~\cite{BGGKO15}.  
Each factor $f_i$ in the LZ factorization
must be as long as possible and must be either the first occurrence of a
letter in $x$ or occur in $f_1\ldots f_{i-1}$.\footnote{This is the
  non-overlapping version of the LZ factorization.}  The Lyndon
factorization, on the other hand, was first studied in the context of
combinatorics on words~\cite[Sect.~5]{Lot97}, and later found
use in algorithms; for example, in a bijective variant of the Burrows--Wheeler transform~\cite{GS12,K09},
in suffix sorting~\cite{MRRS14} and in repetition
detection~\cite{BIINTT15}. Each factor $f_i$ in the Lyndon
factorization must be a Lyndon word: a string that is
lexicographically smaller than all its proper suffixes; and the
factors must be lexicographically
non-increasing.
Lyndon words themselves
have deep combinatorial properties~\cite{Lot97} and have wide
application~\cite{BLPR09,C04,DP07,HMM12,K09,LR95,Lot97,M13}.

For some problems each factorization (Lempel-Ziv or Lyndon) leads to
quite different solutions. Perhaps the best known example of this is
the computation of all the maximal repetitions --- also known as the
``runs'' --- in a string. In 1999 Kolpakov and Kucherov proved that
$\rho(n)$, the number of runs in a string of length $n$, is $O(n)$,
and showed how to exploit the structure of the LZ factorization to
compute all the runs in linear time~\cite{KK99}.  Much more recently,
Bannai et al.~\cite{BIINTT14} used properties of the Lyndon
factorization to obtain a much simpler constructive proof that
$\rho(n) < n$. This later result also leads to a straight-forward
linear-time algorithm for computing the runs from the Lyndon
factorization~\cite{BIINTT15}.

Our overarching motivation in this paper is to obtain a deeper
understanding of how these two fundamental factorizations ---
Lempel-Ziv and Lyndon --- relate. Toward this aim, we ask: 
\emph{by how much can the sizes of the factorizations of the same word differ}? Here the size of the Lempel-Ziv factorization $s=p_1\cdots p_z$ is $z$ and the size of the Lyndon factorization $s=f_1^{e_1} \cdots f_m^{e_m}$, where each $e_i$ is positive and each $f_i$ is lexicographically strictly greater than $f_{i+1}$, is $m$. For most strings, the number of Lyndon
factors is much smaller. Indeed, any string has a rotation with a
Lyndon factorization of size one. So the actual question is how big can $m$ be with respect to $z$. For a lower bound, we show that there are strings with $m = z + \Theta(\sqrt{z})$. Our main result is the upper bound: the inequality $m < 2z$ holds for all strings. 
This result improves significantly a previous, indirect bound by I et
al.~\cite{INIB16}, who showed that the number of Lyndon factors cannot
be more than the size of the smallest straight line program
(SLP). Since the smallest SLP is at most a logarithmic factor bigger
than the LZ factorization~\cite{Ryt03,CLLPPSS05}, this establishes an
indirect, logarithmic factor bound, which we improve to a constant
factor two.

\section{Basic Notions}
\label{sec-notation}

\newcommand{\block}{\mathit{B}}
\newcommand{\lf}[1]{\mathit{LF}({#1})}
\newcommand{\lz}[1]{\mathit{LZ}({#1})}

We consider finite strings over an alphabet
$\Sigma=\{a_1,\ldots,a_n\}$, which is linearly ordered: $a_1\prec
a_2\prec\cdots\prec a_n$.  For strings, we use the array notation:
$s=s[1..|s|]$, where $|s|$ stands for the length of $s$.  The empty
string $\varepsilon$ has length 0. Any pair $i,j$ such that $1\le i\le
j\le |s|$ specifies a \emph{substring} $s[i..j]$ in $s$. A string $u$
equal to some $s[i..j]$ is a \emph{factor} of $s$ (a \emph{prefix}, if
$i=1$, and a \emph{suffix}, if
$j=|s|$). A prefix or suffix of $s$ is called \emph{proper}
if it is not equal to $s$. A factor $u$ may be equal to several substrings of
$s$, referred to as \emph{occurrences} of $u$ in $s$. By $u^k$
we denote the concatenation of $k$ copies of string $u$. If $k=0$ we define
$u^k=\varepsilon$.

A string $u$ over $\Sigma$ is lexicographically smaller or equal than a string
$v$ (denoted by $u\preceq v$) if either $u$ is a prefix of $v$ or
$u=xaw_1$, $v=xbw_2$ for some strings $x,w_1,w_2$ and some letters
$a\prec b$.  In the latter case, we refer to this occurrence of $a$
(resp., of $b$) as the \emph{mismatch} of $u$ with $v$ (resp., of $v$
with $u$). A string $w$ is called a \emph{Lyndon word} if $w$ is
lexicographically smaller than all its non-empty proper suffixes.  The
\emph{Lyndon factorization} of a string $s$ is its unique (see
\cite{CFL58}) factorization $s=f_1^{e_1} \cdots f_m^{e_m}$ such that
each $f_i$ is a Lyndon word, $e_{i} \geq 1$, and $f_{i} \succ f_{i+1}$
for all $1 \leq i < m$.
We call each $f_i$ a \emph{Lyndon factor} of $s$, and each $F_i= f_i^{e_i}$ a \emph{Lyndon run} of $s$. The size of the Lyndon factorization is $m$, the number of distinct Lyndon factors, or equivalently, the number of Lyndon runs.

\clearpage
The \emph{non-overlapping LZ factorization} (see \cite{ZL77}) of a
string $s$ is its factorization $s=p_1 \cdots p_z$ built left to right in a greedy way by the following rule: each new factor (also called an \emph{LZ77 phrase})
$p_i$ is either the leftmost occurrence of a letter in $s$ or the longest prefix of $p_i\cdots p_z$ which occurs in $p_1\cdots p_{i-1}$.

\section{Proof of the Upper Bound}
\label{sec:upper-bound-general}

The aim of this section is to prove the following theorem.

\begin{theorem}
\label{thm:upper-bound}
Every string $s$ having Lyndon factorization $s=f_1^{e_1} \cdots f_m^{e_m}$ and non-overlapping LZ factorization $s=p_1 \cdots p_z$ satisfies $m < 2z$.
\end{theorem}

Let us fix an arbitrary string $s$ and relate all notation ($f_i, e_i, F_i, p_i, m, z$) to $s$. The main line of the proof is as follows. We identify occurrences of some factors in $s$ that must contain a boundary between two LZ77 phrases. Non-overlapping occurrences  contain different boundaries, so our aim is to prove the existence of more than $m/2$ such occurrences. We start with two basic facts; the first one is obvious.

\begin{lemma}
\label{lm:between}
For any strings $u,v,w_1,w_2$, the relation $uw_1\prec v \prec uw_2$ implies that $u$ is a prefix of $v$. 
\end{lemma}

\begin{lemma}
\label{lm:f_vs_F}
The inequality $j<i$ implies $f_j\succ F_i$.
\end{lemma}
\begin{proof}
We prove that $f_j\succ f_i^k$ for any $k$, arguing by induction on $k$. The base case $k=1$ follows from the definitions. Let $f_j\succ f_i^{k-1}$. In the case of mismatch, $f_j\succ f_i^k$ holds trivially.  Otherwise, $f_j=f_i^{k-1}x$ for some $x\ne\varepsilon$. If $x=f_i$ or  $x\prec f_i$, then $x\prec f_j$, and so $f_j$ is not a Lyndon word. Hence $x\succ f_i$ and thus $f_j=f_i^{k-1}x\succ f_i^k$. Thus, the inductive step holds.
\end{proof}

The next lemma locates the leftmost occurrences of the Lyndon runs and their products.

\begin{lemma}
\label{lm:prev-occ}
Let $d \geq 1$ and $1 \leq i \leq m-d+1$, and assume that
$F_i F_{i+1} \cdots F_{i+d-1}$
has an occurrence to the left of the trivial one in $s$.
Then:
\begin{enumerate}
\item The leftmost occurrence of $F_i F_{i+1} \cdots F_{i+d-1}$ is a prefix of $f_j$ for some $j<i$;
\item $F_i F_{i+1}\cdots F_{i+d-1}$
is a prefix of every $f_k$ with $j<k<i$.
\end{enumerate}
\end{lemma}
\begin{proof}
(1) Let $j$ be the smallest integer such that the leftmost occurrence
of $F_i F_{i+1} \cdots F_{i+d-1}$ in $s$ overlaps $F_j$.
Suppose first that the leftmost occurrence of $F_i F_{i+1} \cdots F_{i+d-1}$
is not entirely contained inside
a single occurrence of $f_j$.
Then there exists a non-empty suffix $u$ of $f_j$ that is equal to some prefix of one of
the factors $f_i, \ldots, f_{i+d-1}$, say $f_{i'}$.
We cannot have $u=f_j$ because then $f_j \preceq f_{i'}$ which is impossible since $j<i'$.
Thus $u$ must be a proper suffix of $f_j$. But then $u \preceq f_{i'} \prec f_j$, which
contradicts $f_j$ being a Lyndon word.

Suppose then that the leftmost occurrence of $F_i F_{i+1}\cdots F_{i+d-1}$
in $s$ is entirely contained inside $f_j$ but is not its prefix, i.e.,
$f_j=vF_i F_{i+1} \cdots F_{i+d-1}w$ for some strings $v \neq \varepsilon$
and $w$. Since $f_j$ is a Lyndon word we have
$f_j \prec F_i F_{i+1} \cdots F_{i+d-1}w$.
Consider the position of the
mismatch of $F_i F_{i+1} \cdots F_{i+d-1}w$ with $f_j$.
If the mismatch occurs inside $F_i F_{i+1} \cdots F_{i+d-1}$, we can write
$f_j=F_i \cdots F_{i'-1} f_{i'}^{e}x$ where $i \leq i' < i+d$,
$0 \leq e < e_i$, and $x$ is a suffix of $f_j$ that satisfies $x\prec f_{i'}\prec f_j$,
which contradicts $f_j$ being a Lyndon word.
On the other hand, the mismatch inside $w$ implies that $f_j$ begins with
$F_i F_{i+1} \cdots F_{i+d-1}$, contradicting the assumption that the
inspected occurrence of $F_i F_{i+1} \cdots F_{i+d-1}$ is the leftmost in $s$.

(2) We prove this part by induction on $d$.
Let $d=1$. By Lemma~\ref{lm:f_vs_F} we have $f_j\succ f_k\succ F_i$. Since $f_j$ begins with $F_i$ by statement 1, so does $f_k$ by Lemma~\ref{lm:between}. 
Assume now that the claim holds for all $d'<d$.
From the inductive assumption $F_i$ and $F_{i+1} \cdots F_{i+d-1}$
  are both prefixes of $f_k$.
Let $y, y'$, and $z$ be such that
  $f_{j}=F_{i} F_{i+1} \cdots F_{i+d-1}y$,
  $f_{k}=F_{i+1} \cdots F_{i+d-1}y'=F_{i}z$.
We have $j<k$ and thus $f_{k} \prec f_{j}$ must hold
  which, since $F_{i}$ is a prefix of both $f_{j}$
  and $f_{k}$, implies $z \prec F_{i+1} \cdots F_{i+d-1}y$.
On the other hand, since $f_{k}$ is a Lyndon word, we have
  $f_{k}=F_{i+1} \cdots F_{i+d-1}y' \prec z$.
By Lemma~\ref{lm:between}, $F_{i+1} \cdots F_{i+d-1}y' \prec z \prec F_{i+1} \cdots F_{i+d-1}y$
  implies that $F_{i+1} \cdots F_{i+d-1}$ is a prefix of $z$ or
  equivalently that $F_{i} F_{i+1} \cdots F_{i+d-1}$ is a prefix of $f_{k}$.
\end{proof}

\subsection{Domains}

Lemma~\ref{lm:prev-occ} motivates the following definition.

\begin{definition}
Let $d \geq 1$ and $1 \leq i \leq m-d+1$. We define the
\emph{$d$-domain} of Lyndon run $F_i$ as
the substring $\dom_d(F_i)=F_j F_{j+1} \cdots F_{i-1}$, $j \leq i$
of $s$, where $F_j$ is the Lyndon run (which exists by
Lemma~\ref{lm:prev-occ}) starting at the same position as
the leftmost occurrence of $F_i F_{i+1} \cdots F_{i+d-1}$ in $s$.
Note that if $F_{i} F_{i+1} \cdots F_{i+d-1}$ does not have any
occurrence to the left of the trivial one then $\dom_d(F_i)=\varepsilon$.
The integers $d$ and $i-j$ are called the \emph{order} and \emph{size}
of the domain, respectively.

The \emph{extended $d$-domain} of $F_i$ is the
substring $\extdom_d(F_i)=\dom_d(F_i) F_i \cdots F_{i+d-1}$ of $s$.
\end{definition}

Lemma~\ref{lm:prev-occ} implies two easy properties of domains presented below as Lemma~\ref{prop}.
These properties lead to a convenient graphical notation to
illustrate domains (see Fig.~\ref{fig:example-of-domain}).

\begin{lemma}
\label{prop}
Let $\dom_d(F_i)=F_{j} \cdots F_{i-1}$, $j \leq i$. Then:
\begin{itemize}
\item For any $d'>d$, $\dom_{d'}(F_i)$ is a suffix of $\dom_{d}(F_i)$;
\item For any $d'\geq 1$, $\dom_{d'}(F_k)$ is a substring of $\dom_{d}(F_i)$ if $j \leq k < i$.
\end{itemize}
\end{lemma}

\begin{figure}[t]
\centering
\begin{tikzpicture}[scale=0.27,font=\scriptsize]
  \draw (0,-0.2)  rectangle (3.5,1);
  \draw (1.75,1.2) node[below]{$F_j$};
  \draw (5.25,0.5) node{$\cdots$};
  \draw (7,-0.2) rectangle (10.5,1);
  \draw (8.75,1.2) node[below]{$F_{i-1}$};
  \draw (10.5,-0.2) rectangle (14,1);
  \draw (12.25,1.2) node[below]{$F_i$};
  \draw (15.75,0.5) node{$\cdots$};
  \draw (17.5,-0.2) rectangle (21,1);
  \draw (19.25,1.2) node[below]{$F_{i+d-1}$};
  \draw[-, rounded corners=0.3cm] (10.5, 1) -- (10.5,3) -- (0,3) -- (0, 1);
  \draw (5.25,3) node[above]{$d$};
  \draw (0,-1.4) rectangle (21,-0.2);
  \draw (10.5,-1.7) node[above]{$\extdom_d(F_i)$};
  \draw (0,-2.6) rectangle (10.5,-1.4);
  \draw (5.25,-2.9) node[above]{$\dom_d(F_i)$};
\end{tikzpicture}
\hspace{0.1cm}
\begin{tikzpicture}[scale=0.27,font=\scriptsize]
  \foreach \x/\ch in {-14/a, -13/b, -12/b, -11/a, -10/b, -9/b, -8/a, -7/b, -6/a, -5/b, -4/b, -3/a,
  -2/b, -1/a, 0/b, 1/b, 2/b, 3/a, 4/b, 5/a, 6/b, 7/b, 8/a, 9/b, 10/a}
    \draw (\x+0.5, -0.38) node[above]{$\rm \ch$};
  \draw (-14,-0.3) rectangle (-8, 1);
  \draw (-8,-0.3) rectangle (3, 1);
  \draw (3,-0.3) rectangle (8, 1);
  \draw (8,-0.3) rectangle (10, 1);
  \draw (10,-0.3) rectangle (11, 1);
  \draw[-, rounded corners=0.3cm] (10, 1) -- (10,7.5) -- (-14,7.5) -- (-14, 1);
  \draw (-2,7.3) node[above]{$1$};
  \draw[-, rounded corners=0.3cm] (8, 1) -- (8,6) -- (-14,6) -- (-14, 1);
  \draw (-3,5.8) node[above]{$1$};
  \draw[-, rounded corners=0.3cm] (8, 1) -- (8, 4.5) -- (-8,4.5) -- (-8,1);
  \draw (0,4.3) node[above]{$2$};
  \draw[-, rounded corners=0.3cm] (3,1) -- (3, 3) -- (-8,3) -- (-8, 1);
  \draw (-2.5,2.8) node[above]{$1, 2, 3$};
  \draw (-15.2, -0.53) node[above]{$s[i]$:};
  \draw (-15.2, -1.53) node[above]{$i$:};
  \foreach \x/\ch in {-14/1, -13/2, -12/3, -11/4, -10/5, -9/6, -8/7, -7/8, -6/9, -5/10, -4/11, -3/12,
  -2/13, -1/14, 0/15, 1/16, 2/17, 3/18, 4/19, 5/20, 6/21, 7/22, 8/23, 9/24, 10/25}
    \draw (\x+0.5, -1.53) node[above]{\tiny $ \rm \ch$};
\end{tikzpicture}
\caption{Left: Graphical notation used to illustrate
  $\dom_d(F_i)=F_{j}\cdots F_{i-1}$. Also shown is
  $\extdom_d(F_i)=F_j \cdots F_{i+d-1}$.
  Right: all non-empty domains for the example string
  with the Lyndon factorization of size 5. Note that due to
  Lemma~\ref{prop} there are no non-trivial intersections
  between domains.}
\label{fig:example-of-domain}
\end{figure}
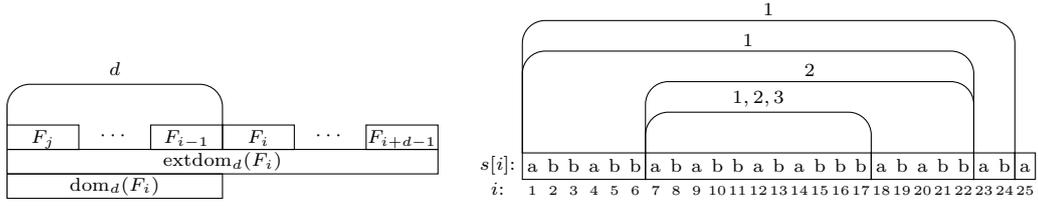

\begin{definition}
\label{def:domain-associated-substring}
Consider $\dom_{d}(F_i)$ for some
$d \geq 1$, $1 \leq i \leq m-d+1$,
and let $\alpha=F_i \cdots F_{i+d-1}$. We say that the leftmost
occurrence of $\alpha$ in $s$ is \emph{associated} with $\dom_d(F_i)$.
\end{definition}

For example, in Fig.~\ref{fig:example-of-domain}, $s[7..9]$ is associated with
$\dom_2(s[23..24])$; $s[7..17]$ is associated with $\dom_1(s[7..17])$ even
though it is not shown, since $\dom_1(s[7..17])=\varepsilon$. Observe that due
to Lemma~\ref{prop} this implies that $\dom_d(s[7..17])=\varepsilon$
for any $d>1$, and hence for example the substring of $s$ associated with
$\dom_2(s[7..17])$ is $s[7..22]$.

A substring $s[i..i{+}k]$, $k \geq 0$ is said to \emph{contain an LZ77 phrase
boundary} if some phrase of the LZ-factorization of $s$ begins in one of the
positions $i,\ldots,i{+}k$. Clearly, non-overlapping substrings contain different
phrase boundaries. Furthermore, if the substring of $s$ does not have any
occurrence to the left (in particular, if it is the leftmost occurrence of
a single symbol), it contains an LZ77 phrase boundary, thus we obtain the
following easy observation.

\begin{lemma}
\label{lm:phrase-in-domain-substr}
Each substring associated with a domain contains an LZ77 phrase
boundary.
\end{lemma}

\subsection{Tandem Domains}

\begin{definition}
Let $d \geq 1$ and $1 \leq i \leq m-d$. A pair of domains
$\dom_{d+1}(F_i)$, $\dom_{d}(F_{i+1})$ is called a \emph{tandem domain}
if $\dom_{d+1}(F_i)\cdot F_i = \dom_{d}(F_{i+1})$ or, equivalently,
if $\extdom_{d+1}(F_i)=\extdom_{d}(F_{i+1})$. Note that we permit
$\dom_{d+1}(F_i)=\varepsilon$.
\end{definition}

For example, $\dom_3(s[18..22])$, $\dom_2(s[23..24])$ is a
tandem domain in Fig.~\ref{fig:example-of-domain}, because we
have  $\extdom_3(s[18..22])=\extdom_2(s[23..24])=s[7..25]$.

\begin{definition}
\label{def:tandem-associated-substring}
Let $\dom_{d+1}(F_{i})$, $\dom_{d}(F_{i+1})$ be a tandem domain.
Since $F_{i+1} \cdots F_{i+d}$ is a prefix of $F_{i}$ by
Lemma~\ref{lm:prev-occ}, we let $F_i=F_{i+1}\cdots F_{i+d}x$.
The leftmost occurrence of $F_{i} \cdots F_{i+d}$ in $s$
can thus be written as $F_{i+1} \cdots F_{i+d} x F_{i+1} \cdots F_{i+d}$. 
We say that this particular occurrence of the factor $xF_{i+1} \cdots F_{i+d}$
is \emph{associated} with the tandem domain $\dom_{d+1}(F_{i})$,
$\dom_{d}(F_{i+1})$. 
\end{definition}

\begin{remark}
Note that the above definition permits $\dom_{d+1}(F_i)=\varepsilon$.
If $\dom_{d+1}(F_i)\neq \varepsilon$,
then $\alpha$, the substring of $s$ associated with $\dom_{d+1}(F_i)$,
$\dom_{d}(F_{i+1})$, is (by Lemma~\ref{lm:prev-occ})
a substring of $F_j$, where $F_j$, $j<i$ is the
leftmost Lyndon run inside $\dom_{d+1}(F_i)$.
Otherwise, $\alpha$ overlaps at least two Lyndon runs.
In both cases, however, $\alpha$ is a substring of $\extdom_{d+1}(F_i)$.
\end{remark}

\begin{lemma}
Each substring associated with a tandem domain contains an LZ77 phrase
boundary.
\end{lemma}
\begin{proof}
Let $\dom_{d+1}(F_i)$, $\dom_{d}(F_{i+1})$ be a tandem domain and let
$u=x F_{i+1} \cdots F_{i+d}$ be the associated substring of $s$.
Suppose to the contrary that $u$ contains no LZ77 phrase boundaries. Then
some LZ-phrase $p_t$ contains $u$ and the letter preceding $u$.
Since we consider a non-overlapping LZ77 variant, the previous occurrence
of $p_t$ in $s$ must be a substring of $p_1\cdots p_{t-1}$. 
Note, however, that
$u$ is preceded in $s$ by the leftmost occurrence of $F_{i+1}\cdots F_{i+d}$,
which is the prefix of $F_j$ (see Definition~\ref{def:tandem-associated-substring}).
Thus, the leftmost occurrence of $u$ in $s$ either immediately precedes the
associated substring, or overlaps it, or coincides with it.
This, however, rules out the possibility that the previous occurrence of $p_t$
occurs in $p_1 \cdots p_{t-1}$, a contradiction.
\end{proof}

We say that a tandem domain $\dom_{d+1}(F_i)$, $\dom_{d}(F_{i+1})$
is \emph{disjoint} from a tandem domain $\dom_{e+1}(F_k)$, $\dom_{e}(F_{k+1})$
if all $i, i+1, k, k+1$ are different, i.e., $i+1<k$ or $k+1<i$.

\begin{lemma}
\label{lm:disjoint-tandems}
Substrings associated with disjoint tandem domains do not overlap each other.
\end{lemma}
\begin{proof}
Let $\dom_{d+1}(F_i)$, $\dom_{d}(F_{i+1})$ and $\dom_{e+1}(F_k)$,
$\dom_{e}(F_{k+1})$ be tandem domains
called
the $d$-tandem and $e$-tandem, respectively.
Without the loss of generality let
$i+1<k$.

Case 1: $\dom_{d+1}(F_i) \neq \varepsilon$ and $\dom_{e+1}(F_k)\neq\varepsilon$.
First observe that if
the $d$-tandem and $e$-tandem begin with different Lyndon runs, then
the associated substrings trivially do not overlap
by the above Remark.
Assume then
that all considered domains start with $F_j$, $j<i$.
By Definition~\ref{def:tandem-associated-substring} we can write $F_j$
as $F_j=F_{i+1}\cdots F_{i+d}x F_{i+1} \cdots F_{i+d}y$, where
$|F_{i+1}\cdots F_{i+d}x|=|F_i|$ and $x F_{i+1}\cdots F_{i+d}$ is
the substring of $s$ associated with the $d$-tandem.
Similarly we have $F_j=F_{k+1}\cdots F_{k+e}x' F_{k+1}\cdots F_{k+e}y'$
where $|F_{k+1}\cdots F_{k+e}x'|=|F_k|$ and $x' F_{k+1}\cdots F_{k+e}$
is the substring of $s$ associated with the $e$-tandem.
However, by Lemma~\ref{lm:prev-occ},
$F_k \cdots F_{k+e}$ is a prefix of $F_{i+1}$ and thus
$|F_{k+1}\cdots F_{k+e}x' F_{k+1}\cdots F_{k+e}| \leq |F_{i+1}|$,
i.e., the substring of $s$ associated with the $e$-tandem
is inside the prefix $F_{i+1}$
of $F_j$ and thus is on the left of the substring associated with
the $d$-tandem.

Case 2: $\dom_{d+1}(F_i)=F_{j}\cdots F_{i-1}$, $j<i$, and
$\dom_{e+1}(F_k)=\varepsilon$. In this case the substring associated with the
$e$-tandem begins in $F_k$ by the above Remark and thus is on the right of the
substring associated with the $d$-tandem.

Case 3: $\dom_{d+1}(F_i)=\varepsilon$ and $\dom_{e+1}(F_k)=\varepsilon$.
This is only possible if $i+d<k$ since otherwise $F_k$ (and thus also
$F_{k+1} \cdots F_{k+e}$) occurs in $F_i$, contradicting
$\dom_{e}(F_{k+1})=F_k$.
Then, $\extdom_{d+1}(F_i)$ does not overlap $\extdom_{e+1}(F_k)$, and
the claim holds by above Remark.

Case 4: $\dom_{d+1}(F_i)=\varepsilon$ and
$\dom_{e+1}(F_k)=F_{j}\cdots F_{k-1}$, $j<k$.
Then, the substring of $s$ associated with $e$-tandem is a substring of $F_j$.
If $i>j$, then clearly $\extdom_{d+1}(F_i)$ does not overlap $F_j$.
On the other hand, if $i<j$, it must also hold $i+d<j$ since
otherwise $F_j$ (and thus also $F_k \cdots F_{k+e}$) occurs in $F_i$,
contradicting $\dom_{e+1}(F_k)=F_j \cdots F_{k-1}$, and thus again,
$\extdom_{d+1}(F_i)$ does not overlap $F_j$. In both cases the Remark above
implies the claim.
Finally, if $i=j$, we must also have $i+1<k$ from the assumption about the
disjointness of $d$- and $e$-tandem.
By Lemma~\ref{lm:prev-occ} we can write
$F_{i}=F_{i+1}\cdots F_{i+d}x$, $F_{i+1}=F_{k}\cdots F_{k+e}x'$ and hence
also $F_{i} \cdots F_{i+d}=F_{k}\cdots F_{k+e}x' F_{i+2}\cdots F_{i+d}x F_{i+1} \cdots F_{i+d}$.
In this decomposition, the substring associated with the $e$-tandem occurs inside
the prefix $F_{k}\cdots F_{k+e}$, and the substring associated with the $d$-tandem
is the suffix $x F_{i+1} \cdots F_{i+d}$, which proves the claim.
\end{proof}

\subsection{Groups}

We now generalize the concept of tandem domain.

\begin{definition}
Let $d \geq 1$, $2 \leq p \leq m$, and $1 \leq i \leq m-d-p+2$. A group
of $p$ domains $\dom_{d+p-1}(F_i)$, $\dom_{d+p-2}(F_{i+1})$, $\ldots$,
$\dom_{d}(F_{i+p-1})$ is called a \emph{$p$-group} if
for all $t=0,\ldots,p-2$ the equality
$\dom_{d+p-1-t}(F_{i+t})\cdot F_{i+t}=\dom_{d+p-2-t}(F_{i+t+1})$ holds or, equivalently, 
$\extdom_{d+p-1}(F_{i})=\ldots=\extdom_{d}(F_{i+p-1})$.
Note that we permit $\dom_{d+p-1}(F_i)=\varepsilon$.
\end{definition}

\begin{figure}[t]
\centering
\begin{tikzpicture}[scale=0.31,font=\footnotesize]
  \draw (0, 0) rectangle (13, 1);
  \draw (18, 0) rectangle (29, 1);
  \draw (29, 0) rectangle (34, 1);
  \draw (34, 0) rectangle (37, 1);
  \draw (37, 0) rectangle (38, 1);
  \draw (34, 1) rectangle (35, 2);
  \draw (29, 1) rectangle (32, 2);
  \draw (32, 1) rectangle (33, 2);
  \draw (29, 2) rectangle (30, 3);
  \draw (18, 1) rectangle (23, 2);
  \draw (23, 1) rectangle (26, 2);
  \draw (26, 1) rectangle (27, 2);
  \draw (18, 2) rectangle (21, 3);
  \draw (21, 2) rectangle (22, 3);
  \draw (18, 3) rectangle (19, 4);
  \draw (23, 2) rectangle (24, 3);
  \draw (0, 1) rectangle (5, 2);
  \draw (5, 1) rectangle (8, 2);
  \draw (8, 1) rectangle (9, 2);
  \draw (5, 2) rectangle (6, 3);
  \draw (0, 2) rectangle (3, 3);
  \draw (3, 2) rectangle (4, 3);
  \draw (0, 3) rectangle (1, 4);
  \draw (6.5, 0.5) node{$u$};
  \draw (15.5, 0.5) node{$\cdots$};
  \draw (23.5, 0.5) node{$v$};
  \draw (31.5, 0.5) node{$w$};
  \draw (35.5, 0.5) node{$x$};
  \draw (37.5, 0.5) node{$y$};
  \draw (2.5, 1.5) node{$w$};
  \draw (6.5, 1.5) node{$x$};
  \draw (8.5, 1.5) node{$y$};
  \draw (1.5, 2.5) node{$x$};
  \draw (3.5, 2.5) node{$y$};
  \draw (0.5, 3.5) node{$y$};
  \draw (5.5, 2.5) node{$y$};
  \draw (20.5, 1.5) node{$w$};
  \draw (19.5, 2.5) node{$x$};
  \draw (18.5, 3.5) node{$y$};
  \draw (24.5, 1.5) node{$x$};
  \draw (26.5, 1.5) node{$y$};
  \draw (21.5, 2.5) node{$y$};
  \draw (23.5, 2.5) node{$y$};
  \draw (30.5, 1.5) node{$x$};
  \draw (32.5, 1.5) node{$y$};
  \draw (29.5, 2.5) node{$y$};
  \draw (34.5, 1.5) node{$y$};
  \draw (-2, 0.5) node{$s:$};
  \draw[dotted] (-1, 0) -- (0, 0);
  \draw[dotted] (-1, 1) -- (0, 1);
  \draw[dotted] (38, 0) -- (39, 0);
  \draw[dotted] (38, 1) -- (39, 1);
  \draw[-, rounded corners=0.3cm] (0, 0) -- (0, 5) -- (29, 5) -- (29, 0);
  \draw[-, rounded corners=0.3cm] (0, 0) -- (0, 6.5) -- (34, 6.5) -- (34, 0);
  \draw[-, rounded corners=0.3cm] (0, 0) -- (0, 8) -- (37, 8) -- (37, 0);
  \draw (14.5, 5) node[above]{3};
  \draw (17, 6.5) node[above]{2};
  \draw (18.5, 8) node[above]{1};
  \draw (1, -1) rectangle (4, 0);
  \draw (4, -1) rectangle (9, 0);
  \draw (2.5, -0.5) node{$\alpha_{xy}$};
  \draw (6.5, -0.5) node{$\alpha_{wx}$};
\end{tikzpicture}

\vspace{0.3cm}

\begin{tikzpicture}[scale=0.31,font=\footnotesize]
  \draw (0, 0) rectangle (17, 1);
  \draw (17, 0) rectangle (25, 1);
  \draw (25, 0) rectangle (29, 1);
  \draw (29, 0) rectangle (31, 1);
  \draw (31, 0) rectangle (32, 1);
  \draw (0, 1) rectangle (8, 2);
  \draw (8, 1) rectangle (12, 2);
  \draw (12, 1) rectangle (14, 2);
  \draw (14, 1) rectangle (15, 2);
  \draw (0, 2) rectangle (4, 3);
  \draw (4, 2) rectangle (6, 3);
  \draw (6, 2) rectangle (7, 3);
  \draw (0, 3) rectangle (2, 4);
  \draw (2, 3) rectangle (3, 4);
  \draw (8, 2) rectangle (10, 3);
  \draw (10, 2) rectangle (11, 3);
  \draw (17, 1) rectangle (21, 2);
  \draw (21, 1) rectangle (23, 2);
  \draw (23, 1) rectangle (24, 2);
  \draw (17, 2) rectangle (19, 3);
  \draw (19, 2) rectangle (20, 3);
  \draw (25, 1) rectangle (27, 2);
  \draw (27, 1) rectangle (28, 2);
  \draw (8.5, 0.5) node{$u$};
  \draw (21, 0.5) node{$v$};
  \draw (4, 1.5) node{$v$};
  \draw (27, 0.5) node{$w$};
  \draw (19, 1.5) node{$w$};
  \draw (10, 1.5) node{$w$};
  \draw (2, 2.5) node{$w$};
  \draw (30, 0.5) node{$x$};
  \draw (26, 1.5) node{$x$};
  \draw (22, 1.5) node{$x$};
  \draw (18, 2.5) node{$x$};
  \draw (13, 1.5) node{$x$};
  \draw (9, 2.5) node{$x$};
  \draw (5, 2.5) node{$x$};
  \draw (1, 3.5) node{$x$};
  \draw (31.5, 0.5) node{$y$};
  \draw (27.5, 1.5) node{$y$};
  \draw (23.5, 1.5) node{$y$};
  \draw (19.5, 2.5) node{$y$};
  \draw (2.5, 3.5) node{$y$};
  \draw (6.5, 2.5) node{$y$};
  \draw (10.5, 2.5) node{$y$};
  \draw (14.5, 1.5) node{$y$};
  \draw (-2, 0.5) node{$s:$};
  \draw[dotted] (-1, 0) -- (0, 0);
  \draw[dotted] (-1, 1) -- (0, 1);
  \draw[dotted] (32, 0) -- (33, 0);
  \draw[dotted] (32, 1) -- (33, 1);
  \draw[-, rounded corners=0.3cm] (0, 0) -- (0, 5) -- (17, 5) -- (17, 0);
  \draw[-, rounded corners=0.3cm] (0, 0) -- (0, 6.5) -- (25, 6.5) -- (25, 0);
  \draw[-, rounded corners=0.3cm] (0, 0) -- (0, 8) -- (29, 8) -- (29, 0);
  \draw (8.5, 5) node[above]{4};
  \draw (12.5, 6.5) node[above]{3};
  \draw (14.5, 8) node[above]{2};
  \draw (3, -1) rectangle (7, 0);
  \draw (7, -1) rectangle (15, 0);
  \draw (15, -1) rectangle (32, 0);
  \draw (5, -0.5) node{$\alpha_{wx}$};
  \draw (11, -0.5) node{$\alpha_{vw}$};
  \draw (23.5, -0.5) node{$\alpha_{uv}$};
\end{tikzpicture}
  \vspace{0.2cm}
  \caption{Illustration of Lemma~\ref{lm:p-groups}. In the examples
    $u$, $v$, $w$, $x$, $y$ are Lyndon runs from the Lyndon factorization
    of $s$. The top figure shows a 3-group: $\dom_3(w)=u\cdots v$,
    $\dom_2(x)=u\cdots vw$,
    $\dom_1(y)=u\cdots vwx$. $\alpha_{wx}$ is a substring associated with
    the tandem domain $\dom_3(w)$, $\dom_2(x)$, and $\alpha_{xy}$ is a
    substring associated with the tandem domain $\dom_2(x)$, $\dom_1(y)$.
    Observe that the substrings associated with tandem domains occur as a
    contiguous substring and in
    reverse order (compared to the order of the corresponding tandem domains in $s$).
    The bottom figure shows a 4-group: $\dom_5(u)=\varepsilon$, $\dom_4(v)=u$,
    $\dom_3(w)=uv$, $\dom_2(x)=uvw$ and demonstrates the case when the
    leftmost domain in a group is empty.
  }
  \label{fig:example-of-p-group}
\end{figure}
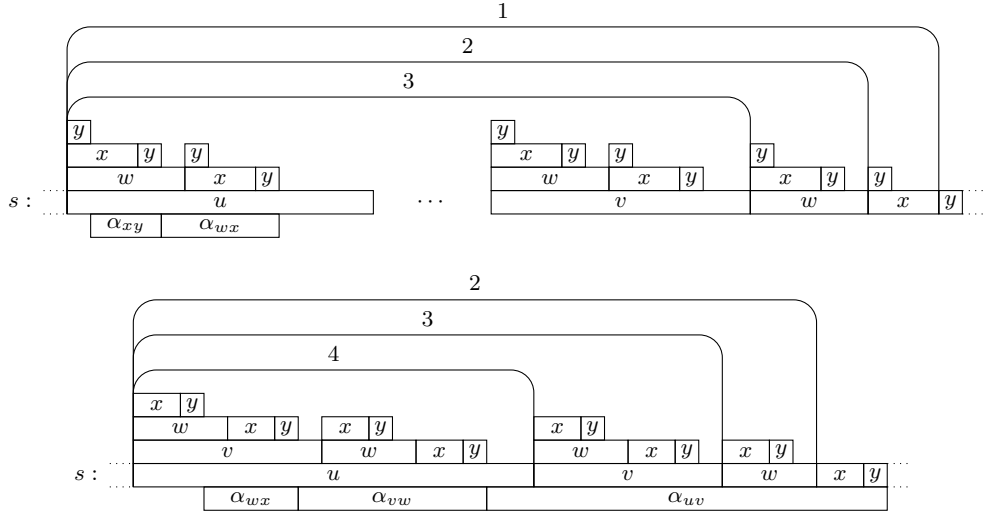

\begin{lemma}
\label{lm:p-groups}
Substrings associated with tandem domains from the same $p$-group
do not overlap each other.
\end{lemma}
\begin{proof}
Assume first that $p=3$.
By Lemma~\ref{lm:prev-occ} we have
$F_i=F_{i+1}\cdots F_{i+d+1}x'$ and $F_{i+1}=F_{i+2}\cdots F_{i+d+1}x$ for
some words $x'$ and $x$. We can thus write the leftmost occurrence of
$F_{i}\cdots F_{i+d+1}$ in $s$ as
$F_{i+2}\cdots F_{i+d+1}x F_{i+2}\cdots F_{i+d+1}x' F_{i+1} \cdots F_{i+d+1}$.
It is easy to see that those occurrences of $x F_{i+2}\cdots F_{i+d+1}$ and
$x' F_{i+1} \cdots F_{i+d+1}$ are associated with (resp.) tandem domains
$\dom_{d+1}(F_{i+1})$, $\dom_{d}(F_{i+2})$ and
$\dom_{d+2}(F_i)$, $\dom_{d+1}(F_{i+1})$,
and thus the claim holds.

For $p>3$ it suffices to consider all subgroups of size three, in 
left-to-right order, to verify that the substrings associated with all
tandem domains occur in reversed order as a contiguous substring
and thus no two substrings overlap each other.
\end{proof}

The above Lemma is illustrated in Fig.~\ref{fig:example-of-p-group}.
It also motivates the following definition which generalizes the
concept of associated substring from tandem domains to $p$-groups.

\begin{definition}
Consider a $p$-group $\dom_{d+p-1}(F_i)$, $\dom_{d+p-2}(F_{i+1})$,
$\ldots$, $\dom_{d}(F_{i+p-1})$ for some $p \geq 2$. From
Lemma~\ref{lm:prev-occ}, $F_{i+p-1}\cdots F_{i+p+d-2}$
is a prefix of $F_{i}$. Thus, the leftmost occurrence of
$F_{i}\cdots F_{i+p+d-2}$ in $s$ can be written as
$F_{i+p-1}\cdots F_{i+p+d-2}x F_{i+1} \cdots F_{i+p+d-2}$.
We say that this particular occurrence of the substring
$x F_{i+1} \cdots F_{i+p+d-2}$ is \emph{associated} with
the $p$-group $\dom_{d+p-1}(F_i)$, $\dom_{d+p-2}(F_{i+1})$,
$\ldots$, $\dom_{d}(F_{i+p-1})$.
\end{definition}

It is easy to derive a formal proof of the following
Lemma from the proof of Lemma~\ref{lm:p-groups}.

\begin{lemma}
\label{lm:concatenation}
The substring associated with a $p$-group is
the concatenation, in reverse order, of the substrings
associated with the tandem domains belonging to the $p$-group.
\end{lemma}

Our consideration of $p$-groups culminates in the next two results.

\begin{corollary}
\label{cor:phrase-in-p-group-substr}
The substring associated with a $p$-group ($p\geq 2$) contains at least
$p-1$ different LZ77 phrase boundaries.
\end{corollary}

We say that a $p$-group
$\dom_{d+p-1}(F_i)$, $\ldots$, $\dom_{d}(F_{i+p-1})$
is \emph{disjoint} from a $p'$-group
$\dom_{d'+p'-1}(F_k)$, $\ldots$, $\dom_{d'}(F_{k+p'-1})$
if $i+p-1<k$ or $k+p'-1<i$. By combining
Lemma~\ref{lm:disjoint-tandems} and Lemma~\ref{lm:concatenation} we
obtain the following fact.

\begin{lemma}
\label{lm:no-overlapping}
Substrings associated with disjoint groups (each of size $\geq 2$)
do not overlap.
\end{lemma}

\subsection{Subdomains}

The concept of $p$-group does not easily extend
to $p=1$. If we simply define the 1-group as a single domain and
extend the notion of groups to include 1-groups then
Lemma~\ref{lm:no-overlapping} no longer holds. Instead, we introduce
a weaker lemma (Lemma~\ref{lm:domain-vs-tandem}) that also includes single domains.

\begin{definition}
\label{def:subdomains}
We say that a domain $\dom_e(F_k)$ is a \emph{subdomain} of a domain
$\dom_d(F_i)=F_j \cdots F_{i-1}$, $j \leq i$ if $k=i$ and $e=d$ (i.e., the
domain is its own subdomain), or $j\leq k < i$ and $\extdom_e(F_k)$ is
a substring of $\extdom_d(F_i)$
(or equivalently, if $k+e \leq i+d$).
In other words, $F_k$ has to be
one of the Lyndon runs among $F_j$, $\ldots$, $F_{i-1}$ and the extended domain of $F_k$
cannot extend (to the right) beyond the extended domain of $F_i$.
\end{definition}

\begin{lemma}
\label{lm:domain-vs-tandem}
Consider a tandem domain $\dom_{e+1}(F_k)$, $\dom_{e}(F_{k+1})$ such that
$\dom_{e+1}(F_k)$ and $\dom_{e}(F_{k+1})$ are
subdomains of $\dom_{d}(F_i)$.
Then, the substring associated with the tandem domain $\dom_{e+1}(F_k)$,
$\dom_{e}(F_{k+1})$ does not overlap the substring associated with
$\dom_{d}(F_i)$.
\end{lemma}
\begin{proof}
First, observe that in order for a tandem domain consisting of two subdomains
to exist, $\dom_{d}(F_i)$ has to be non-empty. Thus, let
$\dom_{d}(F_i)=F_j \cdots F_{i-1}$ for some $j<i$. This implies
(Lemma~\ref{lm:prev-occ}) that
the substring associated with $\dom_{d}(F_i)$ is a prefix of $F_j$.

Assume first that $\dom_{e+1}(F_k)=F_{j'}\cdots F_{k-1}$, $j<j' \leq k$.
The substring associated with the tandem domain
is a substring of $\extdom_{e+1}(F_k)$ thus
it trivially does not overlap $F_j$.

Assume then that $\dom_{e+1}(F_k)=F_j\cdots F_{k-1}$.
If $k+1<i$ then by Lemma~\ref{lm:prev-occ}, $F_{i} \cdots F_{i+d-1}$
is a prefix of $F_{k+1}$. By Definition~\ref{def:tandem-associated-substring} the
leftmost occurrence of $F_{k} \cdots F_{k+e}$ in $s$ can be written as
$F_{k+1} \cdots F_{k+e} x F_{k+1} \cdots F_{k+e}$.
Thus clearly the leftmost occurrence of
$F_{i} \cdots F_{i+d-1}$ (associated with $\dom_k(F_i)$)
occurs in a prefix $F_{k+1}$ not overlapped by $x F_{k+1} \cdots F_{k+e}$
(which is a substring associated with the tandem domain).

The remaining case is when $k+1=i$. Then by Definition~\ref{def:subdomains} we must
have $e=d$ and again the claim holds easily from
Definition~\ref{def:tandem-associated-substring}.
\end{proof}

For any domain $\dom_{d}(F_i)=F_j \cdots F_{i-1}, j < i$ we define
the set of \emph{canonical subdomains} as follows.
Consider the following procedure. Initialize the set of canonical subdomains
to contain $\dom_{d}(F_i)$. Then initialize $\delta=d$ and start
scanning the Lyndon runs $F_{j}$, $\ldots$, $F_{i-1}$ right-to-left. When
scanning $F_t$ we check if $\dom_{\delta+1}(F_t)=F_{j} \cdots F_{t-1}$.
\begin{itemize}
\item If yes,
we include $\dom_{\delta+1}(F_t)$ into the set, increment $\delta$ and continue
scanning from $F_{t-1}$.
\item Otherwise, i.e., if $\dom_{\delta+1}(F_t)=F_{j'} \cdots F_{t-1}$ for
some $j'>j$, we include the domain $\dom_{\delta+1}(F_t)$ into the set. Then we
set $\delta=0$ and continue scanning from $F_{j'-1}$. All domains that were
included into the set of canonical subdomains in this case are called \emph{loose}
subdomains.
\end{itemize}

See Fig.~\ref{fig:example-of-canonical-set} for an example.
The above procedure simply greedily constructs groups of domains,
and whenever the candidate for the next domain in the current group does not
have a domain that starts with $F_j$, we terminate the current group, add the
loose subdomain into the set and continue building groups starting with the next
Lyndon run outside the (just included) loose subdomain.

Note that the current group can be terminated when containing just one domain,
so it is not a $p$-group in this case. Hence we call
the resulting sequences of non-loose domains \emph{clusters}, i.e.,
a cluster is either a single domain, or a $p$-group ($p \geq 2$). 
Note also that during the construction we may encounter more than one loose subdomain in
a row, so clusters and loose subdomains do not necessarily alternate, but
no two clusters occur consecutively.

Finally, observe that the sequence of clusters and loose subdomains always
ends with a cluster (possibly of size one) containing $\dom_{d'}(F_j)$ for
some $d'$ ($d'=3$ for the example in Fig.~\ref{fig:example-of-canonical-set}), since $\dom_{d'}(F_j)=\varepsilon$ for all $d'$.

\subsection{Proof of the Main Theorem}

We are now ready to prove the key Lemma of the proof. Recall that the size
of $\dom_d(F_i)=F_j \cdots F_{i-1}$, $j \leq i$ is defined as $i-j$.

\begin{figure}[t!]
  \centering
  \begin{minipage}{\textwidth}
    \centering
    \begin{tikzpicture}[scale=0.30,font=\footnotesize]
      \draw (0, 0) rectangle (2, 1.3);
      \draw (2, 0) rectangle (4, 1.3);
      \draw (4, 0) rectangle (6, 1.3);
      \draw (6, 0) rectangle (8, 1.3);
      \draw (8, 0) rectangle (10, 1.3);
      \draw (10, 0) rectangle (12, 1.3);
      \draw (12, 0) rectangle (14, 1.3);
      \draw (14, 0) rectangle (16, 1.3);
      \draw (16, 0) rectangle (18, 1.3);
      \draw (18, 0) rectangle (20, 1.3);
      \draw (20, 0) rectangle (22, 1.3);
      \draw (22.2, 0) rectangle (24.2, 1.3);
      \draw (24.2, 0) rectangle (26.2, 1.3);
      \draw (26.2, 0) rectangle (28.2, 1.3);
      \draw (28.2, 0) rectangle (30.2, 1.3);
      \draw (30.2, 0) rectangle (32.2, 1.3);
      \draw[-, rounded corners=0.01cm] (22, 0) -- (22,3) -- (22.2,3) -- (22.2, 0);
      \draw[-, rounded corners=0.3cm] (10, 0) -- (10, 3) -- (6, 3) -- (6, 0);
      \draw[-, rounded corners=0.15cm] (16, 0) -- (16, 3) -- (14, 3) -- (14, 0);
      \draw[-, rounded corners=0.15cm] (2, 0) -- (2, 3) -- (0, 3) -- (0, 0);
      \draw[-, rounded corners=0.3cm] (4, 0) -- (4, 5) -- (0, 5) -- (0, 0);
      \draw[-, rounded corners=0.3cm] (18, 0) -- (18, 9) -- (0, 9) -- (0, 0);
      \draw[-, rounded corners=0.3cm] (20, 0) -- (20, 11) -- (0, 11) -- (0, 0);
      \draw[-, rounded corners=0.3cm] (24.2, 0) -- (24.2, 13) -- (0, 13) -- (0, 0);
      \draw[-, rounded corners=0.3cm] (26.2, 0) -- (26.2, 15) -- (0, 15) -- (0, 0);
      \draw[-, rounded corners=0.3cm] (28.2, 0) -- (28.2, 17) -- (0, 17) -- (0, 0);
      \draw[-, rounded corners=0.3cm] (12, 0) -- (12, 7) -- (0, 7) -- (0, 0);
      \draw[-, rounded corners=0.01cm] (0, 0) -- (-0.4,2) -- (-0.6,2) -- (-0.2, 0);
      \draw (-0.7, 2) node[above]{3};
      \draw (1, 3) node[above]{2};
      \draw (2, 5) node[above]{1};
      \draw (8, 3) node[above]{2};
      \draw (15, 3) node[above]{3};
      \draw (9, 9) node[above]{2};
      \draw (10, 11) node[above]{1};
      \draw (22.1, 3) node[above]{5};
      \draw (12.1, 13) node[above]{4};
      \draw (13.1, 15) node[above]{3};
      \draw (14.1, 17) node[above]{2};
      \draw (6, 7) node[above]{1};
      \draw (1, -0.2) node[above]{$F_1$};
      \draw (3, -0.2) node[above]{$F_2$};
      \draw (5, -0.2) node[above]{$F_3$};
      \draw (7, -0.2) node[above]{$F_4$};
      \draw (9, -0.2) node[above]{$F_5$};
     \draw (11, -0.2) node[above]{$F_6$};
     \draw (13, -0.2) node[above]{$F_7$};
     \draw (15, -0.2) node[above]{$F_8$};
     \draw (17, -0.2) node[above]{$F_9$};
     \draw (19, -0.2) node[above]{$F_{10}$};
     \draw (21, -0.2) node[above]{$F_{11}$};
     \draw (23.1, -0.2) node[above]{$F_{12}$};
     \draw (25.1, -0.2) node[above]{$F_{13}$};
     \draw (27.1, -0.2) node[above]{$F_{14}$};
     \draw (29.1, -0.2) node[above]{$F_{15}$};
     \draw (31.1, -0.2) node[above]{$F_{16}$};
     \draw (6, -1.3) rectangle (14, 0);
     \draw (14, -1.3) rectangle (22.2, 0);
     \draw (22.2, -1.3) rectangle (32.2, 0);
     \draw (10, -1.6) node[above]{$\extdom_2(F_6)$};
     \draw (18, -1.6) node[above]{$\extdom_3(F_9)$};
     \draw (27.1, -1.6) node[above]{$\extdom_5(F_{12})$};
    \end{tikzpicture}
  \end{minipage}
  \caption{An example showing the set of canonical subdomains of $\dom_2(F_{15})$.
    Using notation from Lemma~\ref{lm:final-lemma}, the set has $p=4$ clusters of size
    (left-to-right): $\ell_1=3, \ell_2=1, \ell_3=2, \ell_4=3$, and $t=3$
    loose subdomains: $\dom_2(F_6)=F_4 F_5$, $\dom_3(F_9)=F_8$,
    $\dom_5(F_{12})=\varepsilon$ of size $k_1=2$, $k_2=1$, $k_3=0$.
    Note how the extended domains of loose subdomains do not overlap each other.
    Furthermore, note that $\extdom_2(F_{15})=F_1\cdots F_{16}$ can be
    factorized as $F_1 \cdots F_{\ell_1}$ concatenated with the extended domains. By
    Corollary~\ref{cor:phrase-in-p-group-substr} and Lemmas~\ref{lm:no-overlapping}
    and~\ref{lm:domain-vs-tandem},
  	$F_1 \cdots F_{\ell_1}$ contains $1+\sum_{h=1}^{p}(\ell_h-1)=6$ LZ77 phrase
    boundaries, while the extended domains $\extdom_2(F_6)$, $\extdom_3(F_9)$,
    $\extdom_5(F_{12})$ contain $\sum_{h=1}^{t}(\lceil k_h/2 \rceil +1)=5$ LZ77
    phrase boundaries by Lemma~\ref{lm:final-lemma}.}
  \label{fig:example-of-canonical-set}
\end{figure}
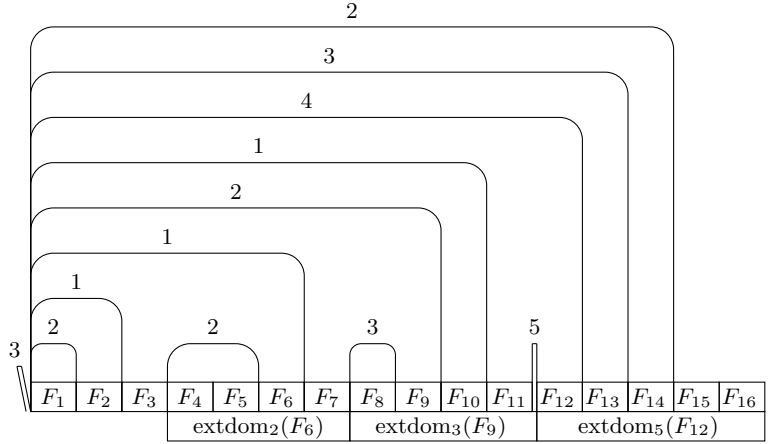

\begin{lemma}
\label{lm:final-lemma}
Let $\dom_d(F_i)$ be a domain of size $k \geq 0$. Then $\extdom_d(F_i)$ contains
at least $\lceil k/2 \rceil + 1$ different LZ77 phrase boundaries.
\end{lemma}
\begin{proof}
Let $\dom_d(F_i)=F_j \cdots F_{i-1}$, $j \leq i$ and $k=i-j$. The proof is by induction
on $k$. For $k=0$, $\extdom_d(F_i)$ is the substring of $s$ associated with
$\dom_d(F_i)$ (see Definition~\ref{def:domain-associated-substring}) and thus
by Lemma~\ref{lm:phrase-in-domain-substr}, $\extdom_d(F_i)$ contains at least
one LZ77 phrase boundary.

Let $k>0$ and assume now that the claim holds for all smaller $k$. Consider
the set ${\cal C}_{i,d}$ of canonical subdomains of $\dom_d(F_i)$. If ${\cal C}_{i,d}$  contains no loose subdomain, it consists of a single cluster which is a $(k+1)$-group. By Corollary~\ref{cor:phrase-in-p-group-substr}, the substring associated with this group contains $k$ phrase boundaries; by Lemma~\ref{lm:domain-vs-tandem}, one more boundary is provided by the domain $\dom_d(F_i)$ itself. We have $1+k\ge 1+\lceil k/2 \rceil$, which concludes the proof of this case.

For the rest of the proof assume that ${\cal C}_{i,d}$ contains $t\ge 1$ loose subdomains, denoted, left to right, by $\dom_{d_1}(F_{i_1}), \ldots, \dom_{d_t}(F_{i_t})$. Note that $d_t>d$. Let $k_h$ be the size of $\dom_{d_h}(F_{i_h})$, $h=1,\ldots,t$. Further, let $\ell\ge 1$ be the size of the leftmost cluster (the one that contains some domain of $F_j$). By the construction of the canonical set we have
\begin{equation}
\label{extdom}
  \extdom_d(F_i)=F_j \cdots F_{j+\ell-1} \extdom_{d_1}(F_{i_1})
    \extdom_{d_2}(F_{i_2})\cdots \extdom_{d_t}(F_{i_t}).
\end{equation}
Both clusters and loose subdomains contribute some number of LZ77 phrase boundaries into their total. The boundaries contributed by clusters are all different by Lemma~\ref{lm:no-overlapping}; let $S$ be their number. These boundaries are also different from the boundary inside the substring associated with $\dom_d(F_i)$ by Lemma~\ref{lm:domain-vs-tandem}. Furthermore, it is easy to see from the proof of Lemma~\ref{lm:domain-vs-tandem} that all these phrase boundaries are located inside $F_j \cdots F_{j+\ell-1}$. The number of phrase boundaries inside the extended domains of loose subdomains can be estimated by the inductive assumption
(by Eq.~\ref{extdom}, these external domains do not overlap each other or $F_j \cdots F_{j+\ell-1}$). So we obtain that $\extdom_d(F_i)$ contains at least
\begin{equation}
\label{equation}
1 + \sum_{h=1}^{t}\left(\left\lceil \frac{k_h}{2} \right\rceil + 1\right) + S
\end{equation}
different LZ77 phrase boundaries. Let us evaluate $\sum_{h=1}^{t} k_h$. By the construction, a loose $d_h$-subdomain is followed by exactly $d_h$ Lyndon runs which are outside loose subdomains; then another loose subdomain follows (cf. Fig.~\ref{fig:example-of-canonical-set}). The only exception is the rightmost loose subdomain, which is followed by $d_t-d$ Lyndon runs outside loose subdomains (note that we only count Lyndon runs inside $\dom_d(F_i)$). Then
\begin{equation}
\label{k_h}
\sum_{h=1}^{t} k_h= k - \ell - \sum_{h=1}^{t} d_h +d.
\end{equation}
Next we evaluate $S$. By Corollary~\ref{cor:phrase-in-p-group-substr}, a cluster of size $r$ contributes $r-1$ phrase boundaries. Then the leftmost (resp., rightmost) cluster contributes $\ell-1$ (resp., $d_t-d-1$) boundaries. Each of the remaining clusters is preceded by a loose $d_h$-subdomain, where $d_h>1$, and contributes $d_h-2$ boundaries. Using Knuth's notation [predicate] for the numerical value (0 or 1) of the predicate in brackets, we can write
\begin{equation}
\label{S}
S= \ell - 1 + \sum_{h=1}^{t} d_h - t - d - \sum_{h=1}^{t-1} [d_h>1].
\end{equation}
Finally, we estimate the number in Eq.~\ref{equation} using Eq. \ref{k_h} and Eq. \ref{S}:
\begin{multline*}
1 + \sum_{h=1}^{t}\left(\left\lceil \frac {k_h}2 \right\rceil + 1\right) + S \ge 1 + t + \frac {k - \ell - \sum_{h=1}^{t} d_h +d}2 + \ell - 1 + \sum_{h=1}^{t} d_h - t - d - \sum_{h=1}^{t-1} [d_h>1]\\
= \frac k2 + \frac {\ell}2 + \sum_{h=1}^{t} \frac {d_h}2 - \frac d2 - \sum_{h=1}^{t-1} [d_h>1] = 
\frac {\ell + d_t - d}2 + \frac{k}{2} + \sum_{h=1}^{t-1} \left(\frac {d_h}2 - [d_h>1]\right) \ge 1 + \frac k2.
\end{multline*}
The obtained lower bound for an integer can be rounded up to $1 + \lceil k/2 \rceil$, as required.
\end{proof}

Using the above Lemma we can finally prove the main Theorem.

\begin{proof}[Proof of Theorem \ref{thm:upper-bound}]
Partition the string $s$ into extended domains as follows: take the string $s'$ such that $s=s' \cdot \extdom_1(F_m)$ and partition $s'$ recursively to get
\[
s= \extdom_1(F_{i_1})\cdots \extdom_1(F_{i_t}), \text{ where } i_t=m.
\]
By Lemma~\ref{lm:final-lemma}, each extended domain $\extdom_1(F_{i_h})$ contains at least $\lceil k_h/2 \rceil+1$ phrase boundaries, where $k_h$ is the size of the domain $\dom_1(F_{i_h})$. Clearly, $\sum_{h=1}^t k_h = m - t$; hence the total number $z$ of the boundaries satisfies
\[
z\ge \sum_{h=1}^t \left(\left\lceil \frac {k_h}2 \right\rceil + 1 \right) \ge \left\lceil \frac {m-t}2 \right\rceil + t =  
\left\lceil \frac {m+t}2 \right\rceil > \frac m2,
\]
as required.
\end{proof}

\section{Lower Bound}

The upper bound on the number of factors in the Lyndon factorization
of a string, given in of Theorem~\ref{thm:upper-bound}, is supported by the
following lower bound. Consider a string $s_k = \block_0 \cdots
\block_k a$, $k \geq 0$, where:
\begin{align*}
  \block_0 &= b, \\
  \block_1 &= ab, \\
  \block_2 &= a^2baba^2b, \\
  \cdots &\\
  \block_k &= (a^kba^1b) \cdots (a^kba^{k-1}b) a^kb.
\end{align*}
For example, $s_3 = (b)(ab)(a^2baba^2b)(a^3baba^3ba^2ba^3b)(a)$.

\begin{theorem}
Let $f_1\cdots f_{m_k}$ and $p_1\cdots p_{z_k}$ be the Lyndon factorization and the non-overlapping LZ factorization of the string $s_k$, $k \geq 2$. Then $m_k = k^2/2 + k/2 + 2$, $z_k = k^2/2 - k/2 + 4$, and thus $m_k=z_k+\Theta(\sqrt{z_k})$.
\end{theorem}

\begin{proof}
First we count Lyndon factors.
All factors will be
different, so their number coincides with the number of Lyndon runs. By the definition of Lyndon factorization, the block $\block_i$ ($0<i\leq k$) is factorized into $i$ Lyndon factors:
\begin{equation} \label{eq:lyndon}
\block_i= a^iba^1b \cdot a^iba^2b \cdots a^iba^{i-1}b \cdot a^ib.
\end{equation}
For any suffix $u$ of $\block_0 \cdots \block_{i-1}$
and any prefix $v$ of $\block_i$, $u \succ v$ holds since $a^i$ is a prefix of $\block_i$ and this is the leftmost occurrence of $a^i$.
Thus there is no Lyndon word that begins in $\block_0 \cdots \block_{i-1}$
and ends in $\block_i$. This implies that the factorization of $s_k$ is the concatenation of the first $b$, then $k$ factorizations Eq. \ref{eq:lyndon}, and the final $a$, $k^2/2 + k/2 + 2$ factors in total.

Let $\lz{s_k}$ denote the LZ factorization of $s_k$. The size of $\lz{s_2}=b\cdot a\cdot ba\cdot aba\cdot baaba$ is 5. For $k\ge 3$, we prove by induction that
\begin{equation} \label{eq:LZ}
\lz{s_k} = \lz{s_{k-1}}\cdot a^{k-1}baba^{k-1} \cdot aba^2ba^{k-1} \cdots aba^{k-2}ba^{k-1} \cdot aba^{k-1}ba^kba.
\end{equation}

For $k=3$ we have $LZ(s_3) = LZ(s_2) \cdot aababaa \cdot abaabaaaba$ and thus the claim holds.
If $k>3$, by the inductive hypothesis the last
phrase in $\lz{s_{k-1}}$ is $p=aba^{k-2}ba^{k-1}ba$. The factor $p$ has only one previous occurrence: it occurs
at the boundary between
$B_{k-2}$ and $B_{k-1}$, followed by $b$. So, $p$ remains a phrase in $\lz{s_k}$.
Each of subsequent $k-2$ phrases of Eq. \ref{eq:LZ} also has a single previous occurrence (inside $B_{k-1}$), and this occurrence is followed by $b$ because $B_{k-1}$ has no factor $a^k$. Thus, Eq. \ref{eq:LZ} correctly represents $\lz{s_k}$. Direct computation now gives $z_k= k^2/2 - k/2 + 4$.
\end{proof}

\bibliography{lzly}

\end{document}